\documentclass[english]{article}
\usepackage[T1]{fontenc}
\usepackage[latin9]{inputenc}
\usepackage{amsthm}
\usepackage{amsmath}
\usepackage{amssymb}
\usepackage{graphicx}
\usepackage[all]{xy}
\usepackage[absolute]{textpos}

\makeatletter
\newcommand{\lyxaddress}[1]{
\par {\raggedright #1
\vspace{1.4em}
\noindent\par}
}
\theoremstyle{plain}
\newtheorem{thm}{\protect\theoremname}
  \theoremstyle{plain}
  \newtheorem{lem}[thm]{\protect\lemmaname}
  \theoremstyle{remark}
  \newtheorem{rem}[thm]{\protect\remarkname}

\newcommand{\se}{\mathsf{s}_{0}}
\newcommand{\so}{\mathsf{s}_{1}}
\newcommand{\cntx}{\mathsf{CNTX}}
\usepackage[all]{xy}
\usepackage{cite}
\usepackage{geometry}
\date{}
\makeatother

\usepackage{babel}
  \providecommand{\lemmaname}{Lemma}
  \providecommand{\remarkname}{Remark}
\providecommand{\theoremname}{Theorem}

\begin{document}

\begin{textblock}{10}(5,1)
\noindent\Large {Foundations of Physics 46, 282--299, 2016}
\end{textblock}

\title{Proof of a Conjecture on Contextuality in Cyclic Systems with Binary
Variables}

\author{Janne V. Kujala\textsuperscript{1} and Ehtibar N. Dzhafarov\textsuperscript{2} }

\maketitle

\lyxaddress{\begin{center}
\textsuperscript{1}University of Jyv\"askyl\"a, jvk@iki.fi\\\textsuperscript{2}Purdue
University, ehtibar@purdue.edu
\par\end{center}}
\begin{abstract}
We present a proof for a conjecture previously formulated by Dzhafarov,
Kujala, and Larsson (Found. Phys.  7, 762-782, 2015).
The conjecture specifies a measure for the degree of contextuality
and a criterion (necessary and sufficient condition) for contextuality
in a broad class of quantum systems. This class includes Leggett-Garg,
EPR/Bell, and Klyachko-Can-Binicioglu-Shumovsky type systems as special
cases. In a system of this class certain physical properties $q_{1},\ldots,q_{n}$
are measured in pairs $\left(q_{i},q_{j}\right)$; every property
enters in precisely two such pairs; and each measurement outcome is
a binary random variable. Denoting the measurement outcomes for a
property $q_{i}$ in the two pairs it enters by $V_{i}$ and $W_{i}$,
the pair of measurement outcomes for $\left(q_{i},q_{j}\right)$ is
$\left(V_{i},W_{j}\right)$. Contextuality is defined as follows:
one computes the minimal possible value $\Delta_{0}$ for the sum
of $\Pr\left[V_{i}\not=W_{i}\right]$ (over $i=1,\ldots,n$) that
is allowed by the individual distributions of $V_{i}$ and $W_{i}$;
one computes the minimal possible value $\Delta_{\min}$ for the sum
of $\Pr\left[V_{i}\not=W_{i}\right]$ across all possible couplings
of (i.e., joint distributions imposed on) the entire set of random
variables $V_{1},W_{1},\ldots,V_{n},W_{n}$ in the system; and the
system is considered contextual if $\Delta_{\min}>\Delta_{0}$ (otherwise
$\Delta_{\min}=\Delta_{0}$). This definition has its justification
in the general approach dubbed Contextuality-by-Default, and it allows
for measurement errors and signaling among the measured properties.
The conjecture proved in this paper specifies the value of $\Delta_{\min}-\Delta_{0}$
in terms of the distributions of the measurement outcomes $\left(V_{i},W_{j}\right)$.
\end{abstract}

\section{Introduction}

According to Acín et al. \cite{Acin2015}, with only few exceptions,
literature on contextuality mostly concerns particular examples and
lacks general theory. Perhaps adding to the list of exceptions, two
recent papers \cite{KDL2015,DKL2015} present a theory of contextuality
that, although not entirely general, applies to a very broad class
of quantum systems. Defining context as the set of physical properties
that are measured conjointly, the novelty of this theory is in that
it applies non-trivially also in the presence of context-dependent
measurement biases (e.g., due to interactions/signaling, or imperfections
in the measurement procedure). When such context-dependent measurement
biases are present, the distribution of the measurement of a given
physical property may vary over different contexts. A system is considered
noncontextual if there exists a joint distribution of the measurements
of all contexts such that the measurements of a given physical property
over different contexts are (in a well-defined sense) maximally correlated.

This definition, formulated in Refs.~\cite{DKL2015,KDL2015}, applies
to all systems where each measurement has a finite number of possible
outcomes. Analytic and computational results, however, are confined
to the subclass of so-called \emph{cyclic system}s, where each physical
property appears in two different contexts, each context consists
of two different physical properties, and each measurement has two
possible outcomes. This class includes Leggett-Garg, EPR/Bell, and
Klyachko-Can-Binicioglu-Shumovsky type systems as special cases. As
the main result of Ref.~\cite{KDL2015}, a criterion (necessary and
sufficient condition) was derived for a system to be contextual given
the joint distributions of the measurements in each context. In Ref.~\cite{DKL2015}
a measure of the degree of contextuality was defined based on how
far the measurements of each physical property over different contexts
are from being maximally correlated. This measure has a theoretically
justified formulation (see below for details) and it can be used to
define a criterion of contextuality: a system is contextual if and
only if the degree of contextuality is positive. Computer-assisted
calculations were used in Ref.~\cite{DKL2015} to derive an expression
for the measure of contextuality for systems of $3$, $4$, or $5$
physical properties. Based on these, a general expression was conjectured
for any number $n\ge2$ of physical properties. Interestingly, the
criterion of contextuality implied by the expression conjectured in
Ref.~\cite{DKL2015} was somewhat simpler and inherently different
in form from the one derived analytically in Ref.~\cite{KDL2015},
although both must be equivalent if the conjecture is true. In this
paper, we show that the conjecture is indeed true, and thereby make
the results complete for the special class of cyclic systems.

\subsection{Terminology and notation}

A \emph{cyclic} (single cycle) system (see Figure 1) is defined as
a system of measured properties $q_{1},\ldots,q_{n}$ ($n\geq2$)
and measurement results (random variables) satisfying the following
conditions \cite{DKL2015,KDL2015}: 
\begin{enumerate}
\item the properties are measured in pairs $\left(q_{1},q_{2}\right),\ldots,\left(q_{n-1},q_{n}\right),\left(q_{n},q_{1}\right)$,
called \emph{contexts}, so that each property enters in precisely
two contexts;
\item the result of measuring $\left(q_{i},q_{i\oplus1}\right)$\emph{,}
$i=1,\ldots,n$, is a pair of jointly distributed $\pm1$ random variables
$\left(V_{i},W_{i\oplus1}\right)$, called a \emph{bunch} (with $\oplus$
denoting circular addition: $i\oplus1=i+1$ for $i<n$, and $n\oplus1=1$).
\end{enumerate}
\begin{figure}
\begin{centering}
\includegraphics[scale=0.35]{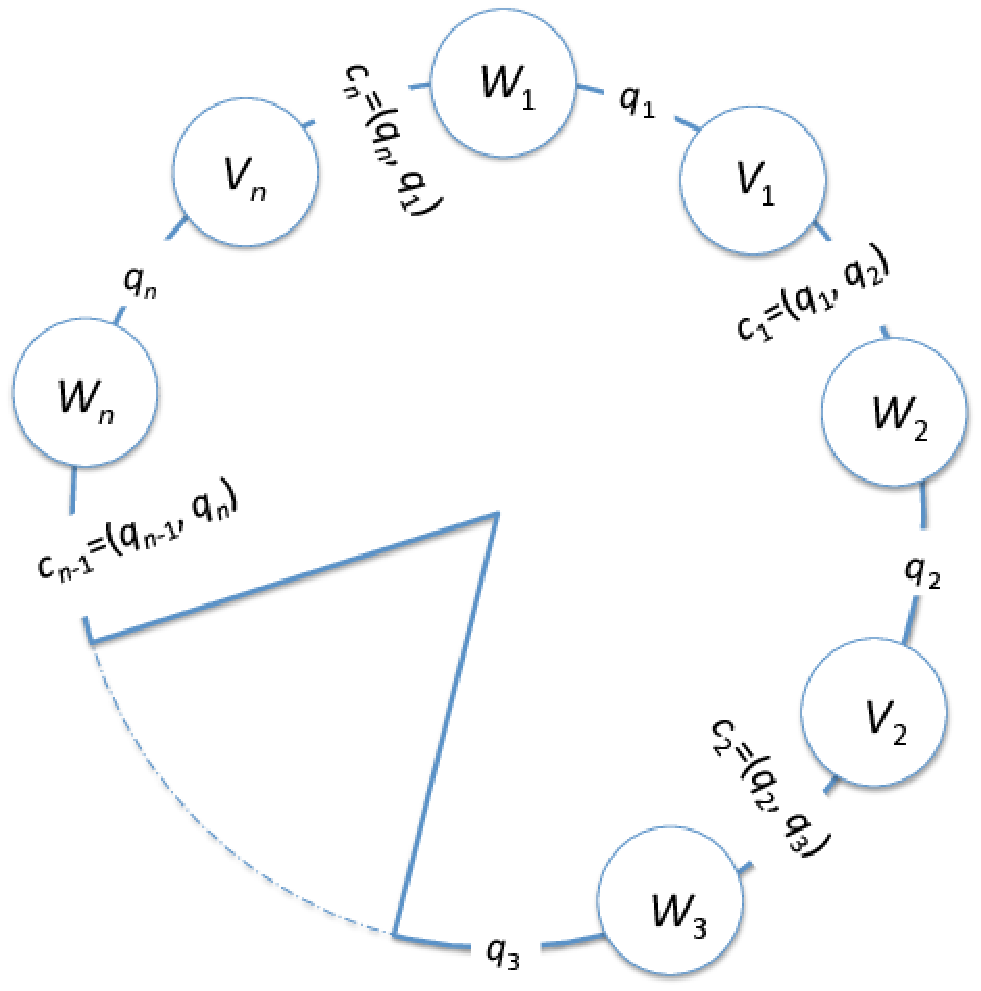}
\par\end{centering}

\protect\caption{A schematic representation of a cyclic system. For $i=1,\ldots,n$,
the measurement of $q_{i}$ is denoted $V_{i}$ if $q_{i}$ is measured
together with $q_{i\oplus1}$; and the measurement of $q_{i}$ is
denoted $W_{i}$ if $q_{i}$ is measured together with $q_{i\ominus1}$.
As a result, the observed pairs of random variables (``bunches'')
are $\left(V_{1},W_{2}\right)$, $\left(V_{2},W_{3}\right)$, ...,
$\left(V_{n},W_{1}\right)$. }
\end{figure}

The Leggett-Garg system \cite{11Leggett,SuppesZanotti1981}, the EPR/Bell
system \cite{10CH,15Fine,9CHSH,Bell1964,Bell1966}, and Klyachko-Can-Binicioglu-Shumovsky
system \cite{Klyachko} are cyclic systems with $n$ equal, respectively,
3, 4, and 5. See \cite{DKL2015,KDL2015} for details.

The distribution of every bunch $\left(V_{i},W_{i\oplus1}\right)$
is uniquely determined by the expectations $\left\langle V_{i}\right\rangle $,
$\left\langle W_{i\oplus1}\right\rangle $, and $\left\langle V_{i}W_{i\oplus1}\right\rangle $.
A pair $\left\{ V_{i},W_{i}\right\} $ of random variables, representing
the same property $q_{i}$ in different contexts, is called a \emph{connection}:
$V_{i},W_{i}$ are not jointly distributed, so the expectation $\left\langle V_{i}W_{i}\right\rangle $
is undefined.

\subsection{Contextuality}

A (probabilistic) coupling for random variables $X,Y,Z,\ldots$ is
defined as any random variable $\left(X^{*},Y^{*},Z^{*},\ldots\right)$
such that $X^{*},Y^{*},Z^{*},\ldots$ are distributed as, respectively
$X,Y,Z,\ldots$. By definition of a random variable (in the broad
sense of the term, including vectors and processes), the components
of $\left(X^{*},Y^{*},Z^{*},\ldots\right)$ are jointly distributed.
For simplicity, we omit asterisks and speak of a coupling $\left(X,Y,Z,\ldots\right)$
for $X,Y,Z,\ldots$ (imposing thereby, non-uniquely, a joint distribution
on $X,Y,Z,\ldots$ that otherwise may not have one). 

In relation to cyclic systems, we are interested in two types of couplings:
(1) couplings $\left(V_{i},W_{i}\right)$ for the connections $\left\{ V_{i},W_{i}\right\} $
($i=1,\ldots,n$); and (2) couplings $\left(\left(V_{1},W_{2}\right),\left(V_{2},W_{3}\right),\ldots,\left(V_{n},W_{1}\right)\right)$
for the set of the observed bunches $\left(V_{1},W_{2}\right),\left(V_{2},W_{3}\right),\ldots,\left(V_{n},W_{1}\right).$
The latter coupling can be written as a random variable $\left(V_{1},W_{2},V_{2},W_{3},\ldots,V_{n},W_{1}\right)$,
with the proviso that its 2-marginals $\left(V_{i},W_{i\oplus1}\right)$
are distributed as the corresponding bunches. Put differently, it
is the coupling for the \emph{entire system of random variables} that
agrees with the observed bunches.

For every $i=1,\ldots,n$, among all couplings for the connection
$\left\{ V_{i},W_{i}\right\} $, we consider one in which $\left\langle V_{i}W_{i}\right\rangle =1-\left|\left\langle V_{i}\right\rangle -\left\langle W_{i}\right\rangle \right|$.
This coupling is called \emph{maximal}, because $1-\left|\left\langle V_{i}\right\rangle -\left\langle W_{i}\right\rangle \right|$
is the maximum possible value for $\left\langle V_{i}W_{i}\right\rangle $
(with given $\left\langle V_{i}\right\rangle $ and $\left\langle W_{i}\right\rangle $).
Equivalently, in the maximal coupling for $\left\{ V_{i},W_{i}\right\} $
the probability $\Pr\left[V_{i}\not=W_{i}\right]$ attains its minimum
possible value: $\frac{1}{2}\left|\left\langle V_{i}\right\rangle -\left\langle W_{i}\right\rangle \right|$.
It is clear that all connections have maximal couplings if and only
if
\begin{equation}
\Delta=\sum_{i=1}^{n}\Pr\left[V_{i}\ne W_{i}\right]
\end{equation}
is at its minimal possible value. We denote this value
\begin{equation}
\Delta_{0}=\frac{1}{2}\sum_{i=1}^{n}\left|\left\langle V_{i}\right\rangle -\left\langle W_{i}\right\rangle \right|.
\end{equation}

Clearly, in every coupling $\left(V_{1},W_{2},V_{2},W_{3},\ldots,V_{n},W_{1}\right)$
for the entire system the value of $\Delta$ is uniquely defined and
cannot fall below $\Delta_{0}$. This leads to the following definition
of a \emph{measure (degree) of contextuality}: it is
\begin{equation}
\cntx=\Delta_{\min}-\Delta_{0}\geq0,\label{eq:first principles}
\end{equation}
where $\Delta_{\min}$ is the smallest possible value of $\Delta$
across all couplings for the entire system. The system is \emph{contextual}
if $\cntx>0$, and it is not if $\cntx=0$. In Ref. \cite{KDL2015},
if a system is not contextual, it is said to have a \emph{maximally
noncontextual description}.

\subsection{Conjecture}

In Ref. \cite{DKL2015}, it was conjectured that
\begin{equation}
\Delta_{\min}=\frac{1}{2}\max\begin{cases}
\so\left(\left\langle V_{i}W_{i\oplus1}\right\rangle :i=1,\dots,n\right)-\left(n-2\right),\\
\sum_{i=1}^{n}\left|\left\langle V_{i}\right\rangle -\left\langle W_{i}\right\rangle \right|,
\end{cases}\label{eq:conjectured Delta_min}
\end{equation}
where the function $\so\left(x_{1},\dots,x_{k}\right)$ for any $k>1$
real-valued arguments is defined as
\begin{equation}
\so\left(x_{1},\dots,x_{k}\right)=\max\sum_{i=1}^{k}m_{i}x_{i},\label{eq:s1}
\end{equation}
with the maximum taken over all $m_{1}\ldots,m_{k}\in\left\{ -1,1\right\} $
such that $\prod_{i=1}^{k}m_{i}=-1$. It follows then that the measure
of contextuality is
\begin{equation}
\cntx=\frac{1}{2}\max\begin{cases}
\so\left(\left\langle V_{i}W_{i\oplus1}\right\rangle :i=1,\dots,n\right)-\sum_{i=1}^{n}\left|\left\langle V_{i}\right\rangle -\left\langle W_{i}\right\rangle \right|-\left(n-2\right),\\
0,
\end{cases}\label{eq:conjectured cntx}
\end{equation}
and a system is contextual if and only if
\begin{equation}
\so\left(\left\langle V_{i}W_{i\oplus1}\right\rangle :i=1,\dots,n\right)>\sum_{i=1}^{n}\left|\left\langle V_{i}\right\rangle +\left\langle W_{i}\right\rangle \right|+\left(n-2\right).\label{eq:conjectured criterion}
\end{equation}
In Ref. \cite{DKL2015} this was shown to be true for $n=3,4,5$,
but not generally. 

In Ref. \cite{KDL2015}, a different criterion for contextuality in
cyclic systems was derived: a system is contextual if and only if
\begin{equation}
\so\left(\left\langle V_{i}W_{i\oplus1}\right\rangle ,1-\left|\left\langle V_{i}\right\rangle -\left\langle W_{i}\right\rangle \right|:i=1,\dots,n\right)>2n-2.\label{eq:proved criterion}
\end{equation}
Of the two criteria, the conjectured (\ref{eq:conjectured criterion})
and the proved (\ref{eq:proved criterion}), the former is more specific,
as it is easy to see that (\ref{eq:proved criterion}) follows from
it (which means that it is known to be a necessary condition for contextuality).
However, this is not the main reason why one should be interested
in (\ref{eq:conjectured criterion}). The main reason is that (\ref{eq:conjectured criterion})
follows from a conjectured formula for the fundamental theoretical
quantity $\Delta_{\min}$ whose excess (\ref{eq:first principles})
over the minimum possible value $\Delta_{0}$ is used as a measure
of contextuality, whereas the degree of violation of (\ref{eq:proved criterion})
does not have any theoretically motivated interpretation as a measure
of contextuality.

\subsection{What we do in this paper}

It is easy to see that (\ref{eq:proved criterion}) and (\ref{eq:conjectured criterion})
are equivalent for \emph{consistently-connected} systems, i.e., those
with $\Delta_{0}=0$. Due to the results obtained in Ref. \cite{DKL2015},
the two criteria should also be equivalent for $n=3,4,5$. However,
it is easy to show by examples that the two inequalities, (\ref{eq:conjectured criterion})
and (\ref{eq:proved criterion}), are not algebraic variants of each
other: in particular, the expressions
\[
\so\left(\left\langle V_{i}W_{i\oplus1}\right\rangle :i=1,\dots,n\right)-\sum_{i=1}^{n}\left|\left\langle V_{i}\right\rangle -\left\langle W_{i}\right\rangle\right| -(n-2)
\]
and
\[
\so\left(\left\langle V_{i}W_{i\oplus1}\right\rangle ,1-\left|\left\langle V_{i}\right\rangle -\left\langle W_{i}\right\rangle \right|:i=1,\dots,n\right)
\]
are not equal to each other. 

Nevertheless, the two inequalities are equivalent, as we prove in
Section \ref{sec:The-equivalence-result} (Theorem \ref{thm:equivalence}).
That is, (\ref{eq:proved criterion}) is indeed a criterion as conjectured
in Ref. \cite{DKL2015}. In Section \ref{sec:Proof-of-the} we prove
the general formula (\ref{eq:conjectured cntx}) for the contextuality
measure (Theorem \ref{thm:measure}). Note that the criterion (\ref{eq:conjectured criterion})
is merely a consequence of the contextuality measure formula (\ref{eq:conjectured cntx}),
i.e., the logical derivability diagram is 
\[
\xymatrix{\textnormal{Definition of Contextuality \eqref{eq:first principles}}\ar[d]\\
\textnormal{Contextuality Measure Formula \eqref{eq:conjectured cntx}}\ar[d]\\
\textnormal{Criterion of Contextuality \eqref{eq:conjectured criterion}}
}
\]
However, in this paper we arrive at the formulas for the measure and
criterion in a more circuitous way. We first prove that the criterion
(\ref{eq:conjectured criterion}) is equivalent to the previously
derived criterion (\ref{eq:proved criterion}), and then we use (\ref{eq:conjectured criterion})
and the definition of contextuality to derive the measure formula
(\ref{eq:conjectured cntx}):

\emph{
\[
\xymatrix@C=2cm{\textnormal{Definition of Contextuality \eqref{eq:first principles}}\ar[r]^{\textnormal{Ref. [2]}}\ar[d]_{_{\textnormal{Theorem }\ref{thm:measure}}} & \textnormal{Criterion of Contextuality \eqref{eq:proved criterion}}\\
\textnormal{Contextuality Measure Formula \eqref{eq:conjectured cntx}} & \textnormal{Criterion of Contextuality \eqref{eq:conjectured criterion}}\ar[l]^{^{\quad\textnormal{Theorem }\ref{thm:measure}}}\ar@{<->}[u]_{_{\textnormal{Theorem }\ref{thm:equivalence}}}
}
\]
}

\section{Results we need for the proofs}

Here we list some results from Ref. \cite{KDL2015}. We make use of
the function $\so$ defined in (\ref{eq:s1}) and the function 
\begin{equation}
\se\left(x_{1},\dots,x_{k}\right)=\max\sum_{i=1}^{k}m_{i}x_{i},\label{eq:s0}
\end{equation}
with $k>1$ real-valued arguments, where the maximum is taken over
all $m_{1}\ldots,m_{k}\in\left\{ -1,1\right\} $ such that $\prod_{i=1}^{k}m_{i}=1$.
\begin{lem}
\label{lem:s0_s1_split}For any $a_{1},\dots,a_{n},b_{1},\dots,b_{m}\in\mathbb{R}$,
\[
\so(a_{1},\dots,a_{n},b_{1},\dots,b_{m})=\max\{\begin{array}[t]{l}
\se(a_{1},\dots,a_{n})+\so(b_{1},\dots,b_{m}),\\
\so(a_{1},\dots,a_{n})+\se(b_{1},\dots,b_{m})\,\}
\end{array}
\]
and
\[
\se(a_{1},\dots,a_{n},b_{1},\dots,b_{m})=\max\{\begin{array}[t]{l}
\se(a_{1},\dots,a_{n})+\se(b_{1},\dots,b_{m}),\\
\so(a_{1},\dots,a_{n})+\so(b_{1},\dots,b_{m})\,\}.
\end{array}
\]

\end{lem}

\begin{lem}
\label{lem:joint-of-two}Jointly distributed $\pm1$-valued random
variables $A$ and $B$ with given expectations $\left\langle A\right\rangle ,\left\langle B\right\rangle ,\left\langle AB\right\rangle $
exist if and only if
\[
\begin{cases}
-1\le\left\langle A\right\rangle \le1,\\
-1\le\left\langle B\right\rangle \le1,\\
\left|\left\langle A\right\rangle +\left\langle B\right\rangle \right|-1\le\left\langle AB\right\rangle \le1-\left|\left\langle A\right\rangle -\left\langle B\right\rangle \right|.
\end{cases}
\]

\end{lem}

\begin{lem}
\label{lem:chain-joint-of-three}Given jointly distributed random
variables $\left(A,B\right)$ and jointly distributed random variables
$\left(B',C\right)$, there exists a coupling $(A,B,B',C)$ such that
$B=B'$ if and only if $B$ and $B'$ have the same distribution.
\end{lem}

\begin{lem}
\label{lem:chain-joint}Jointly distributed $\pm1$ random variables
$A_{1},\dots,A_{n}$ ($n\ge2$) with given expectations $\left\langle A_{1}\right\rangle ,\dots,\left\langle A_{n}\right\rangle $,$\left\langle A_{1}A_{2}\right\rangle ,\dots,\left\langle A_{n-1}A_{n}\right\rangle $
exist if and only if $A=A_{i}$ and $B=A_{i+1}$ satisfy the condition
of Lemma \ref{lem:joint-of-two} for all $i=1,\dots,n-1$.\end{lem}
\begin{thm}
\label{thm:A1-An}Jointly distributed $\pm1$-valued random variables
$A_{1},\dots,A_{n}$ ($n\ge3$) with given expectations 
\[
\left\langle A_{1}\right\rangle ,\dots,\left\langle A_{n}\right\rangle ,\quad\left\langle A{}_{1}A{}_{2}\right\rangle ,\dots,\left\langle A{}_{n-1}A{}_{n}\right\rangle ,\left\langle A{}_{n}A{}_{1}\right\rangle 
\]
exist if and only if $A=A_{i}$ and $B=A_{i\oplus1}$ satisfy the
condition of Lemma \ref{lem:joint-of-two} for all $i=1,\dots,n$
and 
\[
\so\left(\left\langle A{}_{1}A{}_{2}\right\rangle ,\dots,\left\langle A{}_{n-1}A{}_{n}\right\rangle ,\left\langle A{}_{n}A{}_{1}\right\rangle \right)\le n-2.
\]

\end{thm}
The main result of Ref. \cite{KDL2015} is the following theorem.
\begin{thm}
\label{thm:main-criterion}For each $i=1,\dots,n$ ($n\ge2$), let
the distribution of a pair $\left(V_{i},W_{i\oplus1}\right)$ of $\pm1$-valued
random variables be given. Then, the following two statements are
equivalent:
\begin{enumerate}
\item there exists a joint distribution (coupling) of the pairs such that
for all $i=1,\dots,n$ the probability
\[
\Pr\left[V_{i}\ne W_{i}\right]
\]
in the joint is the minimum possible allowed by the marginal distributions
of $V_{i}$ and $W_{i}$ (i.e., there exists a maximally noncontextual
description),
\item the main criterion
\[
\so\left(\left\langle V_{i}W_{i\oplus1}\right\rangle ,1-\left|\left\langle V_{i}\right\rangle -\left\langle W_{i}\right\rangle \right|:i=1,\dots,n\right)\le2n-2
\]
holds true.
\end{enumerate}
\end{thm}

\section{\label{sec:The-equivalence-result}The equivalence result}

In Ref. \cite{DKL2015}, the inequality 
\[
\so\left(\left\langle V_{i}W_{i\oplus1}\right\rangle :i=1,\dots,n\right)\le n-2+\sum_{i=1}^{n}\left|\left\langle V_{i}\right\rangle -\left\langle W_{i}\right\rangle \right|
\]
was derived by computer-assisted calculations as a criterion for the
existence of a maximally noncontextual description for $n=3$, $n=4$,
and $n=5$, and it was conjectured that the same pattern would hold
for all $n\ge2$. In this section, we prove (in Theorem~\ref{thm:equivalence}
below) this conjecture by showing that the inequality shown above
is equivalent to the main criterion of Theorem~\ref{thm:main-criterion}.
\begin{lem}
\label{lem:interesting-trivial}For any numbers $a_{1},\dots,a_{n}\in\mathbb{R}$
($n\ge2)$, exactly one of the following conditions hold:
\begin{enumerate}
\item for some index $k$, the inequality $a_{i}\ge|a_{k}|$ holds for all
$i\ne k$.
\item for some distinct indices $j$ and $k$, the inequality $a_{j}+a_{k}<0$
holds.
\end{enumerate}
\end{lem}

\begin{lem}
\label{lem:expand_s0_s1}Suppose that for some index $k$ the inequalities
$a_{i}\ge a_{k}$ and $a_{i}\ge0$ hold for all $i\ne k$ (this is
implied in particular by condition 1 of Lemma~\ref{lem:interesting-trivial}).
Then,
\[
\so\left(a_{1},\dots,a_{n}\right)=\sum_{i\ne k}a_{i}-a_{k}.
\]
If condition 1 of Lemma \ref{lem:interesting-trivial} holds, then
\[
\se\left(a_{1},\dots,a_{n}\right)=\sum_{i}a_{i}.
\]

\end{lem}

\begin{lem}
\label{lem:WLOG}For each $i=1,\dots,n$ ($n\ge2$), let $\left(V_{i},W_{i\oplus1}\right)$
be a pair of jointly distributed random variables. For any combination
of signs $m_{1},\dots,m_{n}\in\{-1,+1\}$, there exist another set
of random variables
\begin{eqnarray*}
\hat{V}_{i}: & = & m_{i}V_{i},\\
\hat{W}_{i}: & = & m_{i}W_{i}
\end{eqnarray*}
(with each pair $\big(\hat{V}_{i},\hat{W}_{i\oplus1}\big)$ jointly
distributed) preserving the values of $\se$ and $\so$ of the product
expectations:
\begin{align*}
\se\left(\left\langle V_{i}W_{i\oplus1}\right\rangle :i=1,\dots,n\right) & =\se\big(\langle\hat{V}_{i}\hat{W}_{i\oplus1}\rangle:i=1,\dots,n\big),\\
\so\left(\left\langle V_{i}W_{i\oplus1}\right\rangle :i=1,\dots,n\right) & =\so\big(\langle\hat{V}_{i}\hat{W}_{i\oplus1}\rangle:i=1,\dots,n\big),
\end{align*}
and preserving the upper and lower bounds 
\begin{align*}
1-\big|\langle\hat{V}_{i}\rangle-\langle\hat{W}_{i}\rangle\big| & =1-\left|\left\langle V_{i}\right\rangle -\left\langle W_{i}\right\rangle \right|,\\
\big|\langle\hat{V}_{i}\rangle+\langle\hat{W}_{i}\rangle\big|-1 & =\left|\left\langle V_{i}\right\rangle +\left\langle W_{i}\right\rangle \right|-1
\end{align*}
of the hypothetical connection expectation $\left\langle V_{i}W_{i}\right\rangle $
for all $i=1,\dots,n$. This implies in particular that the conditions
(\ref{eq:master}) and (\ref{eq:simpleform}) as well as the measure
(\ref{eq:measure}) to be defined later are all insensitive to negations
of the physical properties.
\end{lem}

\begin{lem}
\label{lem:WLOG-permutation}One can always find such a configuration
of signs $m_{1},\dots,m_{n}\in\{-1,+1\}$ in Lemma~\ref{lem:WLOG}
that $\langle\hat{V}_{i}\hat{W}_{i\oplus1}\rangle$, $i=1,\dots,n$,
satisfy condition 1 of Lemma~\ref{lem:interesting-trivial} and since
the conditions (\ref{eq:master}) and (\ref{eq:simpleform}) as well
as the measure (\ref{eq:measure}) are all symmetrical w.r.t.\ permutations
of the indices, one can generally assume that condition 1 of Lemma~\ref{lem:interesting-trivial}
is satisfied with $k=n$.\end{lem}
\begin{proof}
Let $k$ be the index that minimizes $\left|\left\langle V_{i}W_{i\oplus1}\right\rangle \right|$
and let $m_{k\oplus1}:=1$. Then, define recursively
\[
m_{i\oplus1}:=\begin{cases}
+1, & m_{i}\left\langle V_{i}W_{i\oplus1}\right\rangle \ge0,\\
-1, & \text{otherwise},
\end{cases}
\]
for all $i=k\oplus1,k\oplus2,\dots,k\oplus(n-1)$. This implies $\langle\hat{V}_{i}\hat{W}_{i\oplus1}\rangle=m_{i}m_{i\oplus1}\left\langle V_{i}W_{i\oplus1}\right\rangle \ge0$
for all $i\ne k$ and so $\langle\hat{V}_{i}\hat{W}_{i\oplus1}\rangle=\left|\left\langle V_{i}W_{i\oplus1}\right\rangle \right|\ge\left|\left\langle V_{k}W_{k\oplus k}\right\rangle \right|$
for all $i\ne k$ and condition 1 of Lemma~\ref{lem:interesting-trivial}
is satisfied.\end{proof}
\begin{lem}
\label{lem:inequalities}For any $a,b,c,d\in\mathbb{R}$, we have
\begin{align*}
-\left|d+c\right|+\left|a-c\right|+\left|d-b\right|-\left|a-b\right| & \le2\max\left\{ \left|b\right|,\left|d\right|\right\} ,\\
-\left|a-b\right|-\left|d-c\right|+\left|\left|a-c\right|-\left|d-b\right|\right| & \le0.
\end{align*}
\end{lem}
\begin{proof}
Using the triangle inequality $|a-c|\le|a-b|+|b-(-d)|+|(-d)-c|$,
we obtain
\begin{align*}
 & -|d+c|+|a-c|+|d-b|-|a-b|\\
 & \le-|d+c|+\left(|a-b|+|b-(-d)|+|(-d)-c|\right)+|d-b|-|a-b|\\
 & =|b+d|+|d-b|=2\max\left\{ |b|,|d|\right\} ,
\end{align*}
which is the first inequality. The latter inequality follows as the
conjunction of the following two triangle inequalities
\begin{align*}
\left|a-c\right| & \le\left|a-b\right|+\left|d-c\right|+\left|d-b\right|,\\
\left|d-b\right| & \le\left|a-b\right|+\left|d-c\right|+\left|a-c\right|.
\end{align*}
\end{proof}
\begin{lem}
\label{lem:trivialcase}For all $i=1,\dots,n$ ($n\ge2$), let $\left(V_{i},W_{i}\right)$
be a pair of jointly distributed $\pm1$-valued random variables satisfying
$\left\langle V_{i}W_{i}\right\rangle =1-\left|\left\langle V_{i}\right\rangle -\left\langle W_{i}\right\rangle \right|$
and let $\rho_{1},\dots,\rho_{n}$ satisfy

\begin{equation}
\left|\left\langle V_{i}\right\rangle +\left\langle W_{i\oplus1}\right\rangle \right|-1\le\rho_{i}\le1-\left|\left\langle V_{i}\right\rangle -\left\langle W_{i\oplus1}\right\rangle \right|\label{eq:rho_range}
\end{equation}
for $i=1,\dots,n$. If there exist distinct indices $j$ and $k$
such that
\begin{equation}
\big(1-\big|\langle V_{j}\rangle-\langle W_{j}\rangle\big|\big)+\big(1-\big|\langle V_{k}\rangle-\langle W_{k}\rangle\big|\big)\le0,\label{eq:trivialcasecondition}
\end{equation}
then, there exists a joint distribution of the pairs $\left(V_{i},W_{i}\right)$,
$i=1,\dots,n$, satisfying $\left\langle V_{i}W_{i\oplus1}\right\rangle =\rho_{i}$
for $i=1,\dots,n$.\end{lem}
\begin{proof}
We prove the statement by induction. For $n=2$, denoting $\left\langle V_{i}W_{i\oplus1}\right\rangle =\rho_{i}$
and $\left\langle V_{i}W_{i}\right\rangle =1-\left|\left\langle V_{i}\right\rangle -\left\langle W_{i}\right\rangle \right|$
for $i=1,2$, Theorem~\ref{thm:A1-An} implies that a joint exists
if

\[
\so\left(\left\langle V_{1}W_{2}\right\rangle ,\left\langle V_{2}W_{1}\right\rangle ,\left\langle V_{1}W_{1}\right\rangle ,\left\langle V_{2}W_{2}\right\rangle \right)\le2,
\]
where, without loss of generality (by Lemmas~\ref{lem:WLOG} and
\ref{lem:WLOG-permutation}), we assume $\left\langle V_{1}W_{2}\right\rangle \ge\left|\left\langle V_{2}W_{1}\right\rangle \right|$,
and condition (\ref{eq:trivialcasecondition}) yields $\left\langle V_{1}W_{1}\right\rangle +\left\langle V_{2}W_{2}\right\rangle \le0$.
Thus, we have by Lemma~\ref{lem:s0_s1_split}
\begin{align*}
 & \so\left(\left\langle V_{1}W_{2}\right\rangle ,\left\langle V_{2}W_{1}\right\rangle ,\left\langle V_{1}W_{1}\right\rangle ,\left\langle V_{2}W_{2}\right\rangle \right)\\
 & =\max\begin{cases}
\so\left(\left\langle V_{1}W_{2}\right\rangle ,\left\langle V_{2}W_{1}\right\rangle \right)+\se\left(\left\langle V_{1}W_{1}\right\rangle ,\left\langle V_{2}W_{2}\right\rangle \right)\\
\se\left(\left\langle V_{1}W_{2}\right\rangle ,\left\langle V_{2}W_{1}\right\rangle \right)+\so\left(\left\langle V_{1}W_{1}\right\rangle ,\left\langle V_{2}W_{2}\right\rangle \right)
\end{cases}\\
 & =\max\begin{cases}
\left\langle V_{1}W_{2}\right\rangle -\left\langle V_{2}W_{1}\right\rangle -\left\langle V_{1}W_{1}\right\rangle -\left\langle V_{2}W_{2}\right\rangle \\
\left\langle V_{1}W_{2}\right\rangle +\left\langle V_{2}W_{1}\right\rangle +\left|\left\langle V_{1}W_{1}\right\rangle -\left\langle V_{2}W_{2}\right\rangle \right|
\end{cases}\\
 & \le\max\begin{cases}
\left(1-\left|\left\langle V_{1}\right\rangle -\left\langle W_{2}\right\rangle \right|\right)-\left(\left|\left\langle V_{2}\right\rangle +\left\langle W_{1}\right\rangle \right|-1\right)\\
\qquad-\left(1-\left|\left\langle V_{1}\right\rangle -\left\langle W_{1}\right\rangle \right|\right)-\left(1-\left|\left\langle V_{2}\right\rangle -\left\langle W_{2}\right\rangle \right|\right)\\
\left(1-\left|\left\langle V_{1}\right\rangle -\left\langle W_{2}\right\rangle \right|\right)+\left(1-\left|\left\langle V_{2}\right\rangle -\left\langle W_{1}\right\rangle \right|\right)\\
\qquad+\left|\left(1-\left|\left\langle V_{1}\right\rangle -\left\langle W_{1}\right\rangle \right|\right)-\left(1-\left|\left\langle V_{2}\right\rangle -\left\langle W_{2}\right\rangle \right|\right)\right|
\end{cases}\\
 & =\max\begin{cases}
-\left|\left\langle V_{2}\right\rangle +\left\langle W_{1}\right\rangle \right|+\left|\left\langle V_{1}\right\rangle -\left\langle W_{1}\right\rangle \right|+\left|\left\langle V_{2}\right\rangle -\left\langle W_{2}\right\rangle \right|-\left|\left\langle V_{1}\right\rangle -\left\langle W_{2}\right\rangle \right|\\
2-\left|\left\langle V_{1}\right\rangle -\left\langle W_{2}\right\rangle \right|-\left|\left\langle V_{2}\right\rangle -\left\langle W_{1}\right\rangle \right|+\left|\left|\left\langle V_{1}\right\rangle -\left\langle W_{1}\right\rangle \right|-\left|\left\langle V_{2}\right\rangle -\left\langle W_{2}\right\rangle \right|\right|
\end{cases}\\
 & \le\max\begin{cases}
2\max\left\{ \left|\left\langle W_{2}\right\rangle \right|,\left|\left\langle V_{2}\right\rangle \right|\right\} \\
2
\end{cases}\le2,
\end{align*}
where the two inequalities follow from respectively Lemma~\ref{lem:joint-of-two}
and Lemma~\ref{lem:inequalities}. Thus, the statement holds for
$n=2.$

Assuming then that the statement holds for all systems smaller than
$n$ ($n\ge3$), we prove it for a system of size $n$. With no loss
of generality, assume $j=1<k<n$. By Lemma~\ref{lem:chain-joint},
there exists a joint for the chain $\left(V_{k},W_{k+1},\dots,V_{n},W_{1}\right)$
satisfying $\left\langle V_{i}W_{i\oplus1}\right\rangle =\rho_{i}$
for $i=k,\dots,n$ and this joint has a certain marginal of $\left(V_{k},W_{1}\right)$
satisfying by Lemma~\ref{lem:joint-of-two} the range 
\[
\left|\left\langle V_{k}\right\rangle +\left\langle W_{1}\right\rangle \right|-1\le\left\langle V_{k}W_{1}\right\rangle \le1-\left|\left\langle V_{k}\right\rangle -\left\langle W_{1}\right\rangle \right|.
\]
By the induction assumption, there exists a joint of $\left(V_{1},W_{2},\dots,V_{k},W_{1}\right)$
satisfying $\left\langle V_{i}W_{i\oplus1}\right\rangle =\rho_{i}$
for $i=1,\dots,k-1$ and matching the marginal of $\left(V_{k},W_{1}\right)$
mentioned above. By Lemma~\ref{lem:chain-joint-of-three}, it follows
that there exists a joint of all $\left(V_{i},W_{i}\right)$ satisfying
$\left\langle V_{i}W_{i\oplus1}\right\rangle =\rho_{i}$ for $i=1,\dots,n.$\end{proof}
\begin{thm}
\label{thm:equivalence}For each $i=1,\dots,n$ ($n\ge2$), let the
distribution of a pair $\left(V_{i},W_{i\oplus1}\right)$ of $\pm1$-valued
random variables be given. Then, the main criterion
\end{thm}
\begin{equation}
\so\left(\left\langle V_{i}W_{i\oplus1}\right\rangle ,1-\left|\left\langle V_{i}\right\rangle -\left\langle W_{i}\right\rangle \right|:i=1,\dots,n\right)\le2n-2\label{eq:master}
\end{equation}
is equivalent to
\begin{equation}
\so\left(\left\langle V_{i}W_{i\oplus1}\right\rangle :i=1,\dots,n\right)\le n-2+\sum_{i=1}^{n}\left|\left\langle V_{i}\right\rangle -\left\langle W_{i}\right\rangle \right|.\label{eq:simpleform}
\end{equation}

\begin{proof}
First, without loss of generality (by Lemmas~\ref{lem:WLOG} and
\ref{lem:WLOG-permutation}), we assume $\left|\left\langle V_{n}W_{1}\right\rangle \right|\le\left\langle V_{i}W_{i\oplus1}\right\rangle $,
$i=1,\dots,n-1$. By Lemma~\ref{lem:expand_s0_s1}, this condition
implies 
\begin{eqnarray}
\se\left(\left\langle V_{i}W_{i\oplus1}\right\rangle :i=1,\dots,n\right) & = & \sum_{i=1}^{n}\left\langle V_{i}W_{i\oplus1}\right\rangle ,\label{eq:s0corr}\\
\so\left(\left\langle V_{i}W_{i\oplus1}\right\rangle :i=1,\dots,n\right) & = & \sum_{i=1}^{n-1}\left\langle V_{i}W_{i\oplus1}\right\rangle -\left\langle V_{n}W_{1}\right\rangle .\label{eq:s1corr}
\end{eqnarray}
By Lemma~\ref{lem:s0_s1_split}, the criterion (\ref{eq:master})
is equivalent to the conjunction of
\begin{align}
\so\left(\left\langle V_{i}W_{i\oplus1}\right\rangle :i=1,\dots,n\right)+\se\left(1-\left|\left\langle V_{i}\right\rangle -\left\langle W_{i}\right\rangle \right|:i=1,\dots,n\right) & \le2n-2,\label{eq:s1s0}\\
\se\left(\left\langle V_{i}W_{i\oplus1}\right\rangle :i=1,\dots,n\right)+\so\left(1-\left|\left\langle V_{i}\right\rangle -\left\langle W_{i}\right\rangle \right|:i=1,\dots,n\right) & \le2n-2.\label{eq:s0s1}
\end{align}

\textbf{Case 1.} Suppose now that that the terms $1-\left|\left\langle V_{i}\right\rangle -\left\langle W_{i}\right\rangle \right|$
satisfy for some $k$ the condition $1-\left|\left\langle V_{i}\right\rangle -\left\langle W_{i}\right\rangle \right|\ge\left|1-\left|\left\langle V_{k}\right\rangle -\left\langle W_{k}\right\rangle \right|\right|$
for all $i\ne k.$ Then, Lemma~\ref{lem:expand_s0_s1} yields
\begin{align}
\se\left(1-\left|\left\langle V_{i}\right\rangle -\left\langle W_{i}\right\rangle \right|:i=1,\dots,n\right) & =n-\sum_{i=1}^{n}\left|\left\langle V_{i}\right\rangle -\left\langle W_{i}\right\rangle \right|,\label{eq:s0conn}\\
\so\left(1-\left|\left\langle V_{i}\right\rangle -\left\langle W_{i}\right\rangle \right|:i=1,\dots,n\right) & =n-2-\sum_{i\ne k}\left|\left\langle V_{i}\right\rangle -\left\langle W_{i}\right\rangle \right|+\left|\left\langle V_{k}\right\rangle -\left\langle W_{k}\right\rangle \right|.\label{eq:s1conn}
\end{align}
Now (\ref{eq:s0conn}) implies that (\ref{eq:s1s0}) is equivalent
to (\ref{eq:simpleform}). Furthermore, using (\ref{eq:s1conn}) and
(\ref{eq:s0corr}), the left side of (\ref{eq:s0s1}) becomes 
\begin{align*}
 & \sum_{i=1}^{n}\left\langle V_{i}W_{i\oplus1}\right\rangle +n-2-\sum_{i\ne k}\left|\left\langle V_{i}\right\rangle -\left\langle W_{i}\right\rangle \right|+\left|\left\langle V_{k}\right\rangle -\left\langle W_{k}\right\rangle \right|\\
 & \le2n-2-\sum_{i=1}^{n}\left|\left\langle V_{i}\right\rangle -\left\langle W_{i\oplus1}\right\rangle \right|-\sum_{i\ne k}\left|\left\langle V_{i}\right\rangle -\left\langle W_{i}\right\rangle \right|+\left|\left\langle V_{k}\right\rangle -\left\langle W_{k}\right\rangle \right|
\end{align*}
which, by the triangle inequality
\begin{equation}
\left|\left\langle V_{k}\right\rangle -\left\langle W_{k}\right\rangle \right|\le\sum_{i=1}^{n}\left|\left\langle V_{i}\right\rangle -\left\langle W_{i\oplus1}\right\rangle \right|+\sum_{i\ne k}\left|\left\langle V_{i}\right\rangle -\left\langle W_{i}\right\rangle \right|\label{eq:cyclic-triangle}
\end{equation}
implies (\ref{eq:s0s1}) (i.e., (\ref{eq:s0s1}) always holds in Case
1). It follows that the two conditions for maximal noncontextuality
are equivalent under the assumption of this case.

\textbf{Case 2. }Suppose then that the assumption of Case 1 does not
hold. Then, by Lemma~\ref{lem:interesting-trivial}, the condition
(\ref{eq:trivialcasecondition}) of Lemma~\ref{lem:trivialcase}
holds and so Lemma~\ref{lem:trivialcase} implies that a joint exists,
which, by Theorem~\ref{thm:A1-An} implies that (\ref{eq:master})
holds. However, the condition (\ref{eq:trivialcasecondition}) also
yields $\big|\langle V_{j}\rangle-\langle W_{j}\rangle\big|+\left|\left\langle V_{k}\right\rangle -\left\langle W_{k}\right\rangle \right|\ge2$
which implies that that the right side of (\ref{eq:simpleform}) is
at least $n$ whereas the left side cannot exceed $n$ and so (\ref{eq:simpleform})
holds true as well. Thus, the two conditions for the existence of
a maximally noncontextual description are equivalent under this case
as well.
\end{proof}

\section{\label{sec:Proof-of-the}Proof of the conjecture on the measure of
contextuality}

In this section, we prove the following theorem, which was shown to
be correct for the special cases of $n=3$, $n=4$, and $n=5$ and
conjectured to hold for all $n\ge2$ in Ref.~\cite{DKL2015}:
\begin{thm}
\label{thm:measure}For each $i=1,\dots,n$ ($n\ge2$), let the distribution
of a pair $\left(V_{i},W_{i\oplus1}\right)$ of $\pm1$-valued random
variables be given. Then, the minimum possible value of 
\[
\Delta=\sum_{i=1}^{n}\Pr\left[V_{i}\ne W_{i}\right]
\]
over all possible joints of the given pairs is
\[
\Delta_{\min}=\frac{1}{2}\max\begin{cases}
\so\left(\left\langle V_{i}W_{i\oplus1}\right\rangle :i=1,\dots,n\right)-\left(n-2\right),\\
\sum_{i=1}^{n}\left|\left\langle V_{i}\right\rangle -\left\langle W_{i}\right\rangle \right|.
\end{cases}
\]
\end{thm}
\begin{rem}
Here the bottom expression for $\Delta_{\min}$ corresponds to the
case that a maximally noncontextual description exists, i.e., that
all probabilities $\Pr\left[V_{i}\ne W_{i}\right]$, $i=1,\dots,n$,
are at their minimum values allowed by the marginals. In Ref.~\cite{DKL2015},
the excess of $\Delta_{\min}$ over its minimum possible value $\frac{1}{2}\sum_{i=1}^{n}\left|\left\langle V_{i}\right\rangle -\left\langle W_{i}\right\rangle \right|$
given the marginals is defined as a measure of contextuality.
\end{rem}
To prove Theorem~\ref{thm:measure}, we will first rewrite it in
a more accessible form for our proof. Noting that for a pair of $\pm1$-valued
random variables, the probability $\Pr\left[V_{i}\ne W_{i}\right]$
is fully determined by the expectation $\left\langle V_{i}W_{i}\right\rangle $
through the identity 
\begin{equation}
\left\langle V_{i}W_{i}\right\rangle =1-2\Pr\left[V_{i}\ne W_{i}\right],\label{eq:corr-delta-identity}
\end{equation}
we see that minimizing the probability $\Pr\left[V_{i}\ne W_{i}\right]$
is equivalent to maximizing the connection expectation $\left\langle V_{i}W_{i}\right\rangle $.
Hence, a description is maximally noncontextual if and only if the
connection expectations satisfy 
\[
\left\langle V_{i}W_{i}\right\rangle =1-\left|\left\langle V_{i}\right\rangle -\left\langle W_{i}\right\rangle \right|,\quad i=1,\dots,n,
\]
that is, all the connection expectations are at their maximum values
as given by Lemma~\ref{lem:joint-of-two}. Using (\ref{eq:corr-delta-identity}),
we can rewrite Theorem~\ref{thm:measure} in the following equivalent
form:
\begin{thm}
\label{thm:measure-corr}For each $i=1,\dots,n$ ($n\ge2$), let the
distribution of a pair $\left(V_{i},W_{i\oplus1}\right)$ of $\pm1$-valued
random variables be given. The maximum possible value of 
\[
S=\sum_{i=1}^{n}\left\langle V_{i}W_{i}\right\rangle 
\]
over all possible joints of the given pairs is
\[
M:=\min\begin{cases}
2n-2-\so\left(\left\langle V_{i}W_{i\oplus1}\right\rangle :i=1,\dots,n\right),\\
n-\sum_{i=1}^{n}\left|\left\langle V_{i}\right\rangle -\left\langle W_{i}\right\rangle \right|.
\end{cases}
\]

\end{thm}
The proof needs several lemmas:
\begin{lem}
\label{lem:s0_plus_s1}Suppose $a_{1},\dots,a_{n}\in[-1,1]$. Then,
$\se(a_{1},\dots,a_{n})+\so(a_{1},\dots,a_{n})\le2n-2$.
\end{lem}

\begin{lem}
\label{lem:chain-joint-range}Jointly distributed $\pm1$-valued random
variables $A_{1},\dots,A_{n}$ ($n\ge2$) with expectations $\left\langle A_{1}\right\rangle ,\dots,\left\langle A_{n}\right\rangle ,\left\langle A_{1}A_{2}\right\rangle ,\dots,\left\langle A_{n-1}A_{n}\right\rangle ,\left\langle A_{n}A_{1}\right\rangle $
exist if and only if $A=A_{i}$ and $B=A_{i\oplus1}$ satisfy the
condition of Lemma \ref{lem:joint-of-two} for all $i=1,\dots,n$
and
\[
\se\left(\left\langle A_{1}A_{2}\right\rangle ,\dots,\left\langle A_{n-1}A_{n}\right\rangle \right)-\left(n-2\right)\le\left\langle A_{n}A_{1}\right\rangle \le(n-2)-\so\left(\left\langle A_{1}A_{2}\right\rangle ,\dots,\left\langle A_{n-1}A_{n}\right\rangle \right).
\]
Furthermore, the range given by the above inequalities is always nonempty
and intersects the range
\[
\left|\left\langle A_{n}\right\rangle +\left\langle A_{1}\right\rangle \right|-1\le\left\langle A_{n}A_{1}\right\rangle \le1-\left|\left\langle A_{n}\right\rangle -\left\langle A_{1}\right\rangle \right|
\]
given by Lemma~\ref{lem:joint-of-two}.\end{lem}
\begin{proof}
The first part is a direct corollary of Theorem~\ref{thm:A1-An}
and Lemma~\ref{lem:s0_s1_split}. The fact that the two ranges intersect
follows from the fact that a joint always exists for the chain with
expectations $\left\langle A_{1}A_{2}\right\rangle ,\dots,\left\langle A_{n-1}A_{n}\right\rangle $
and this joint yields some $2$-marginal for $(A_{n},A_{1})$ which
satisfies Lemma~\ref{lem:joint-of-two}. \end{proof}
\begin{lem}
\label{lem:S=00003DM}For each $i=1,\dots,n$ ($n\ge2$), let the
distribution of a pair $\left(V_{i},W_{i\oplus1}\right)$ of $\pm1$-valued
random variables be given. Suppose that the main criterion (\ref{eq:master})
does not hold and formally define the expressions
\begin{align*}
S & :=\sum_{i=1}^{n}\left\langle V_{i}W_{i}\right\rangle ,\\
M & :=2n-2-\so\left(\left\langle V_{i}W_{i\oplus1}\right\rangle :i=1,\dots,n\right).
\end{align*}
If there exists some values $\left\{ \left\langle V_{i}W_{i}\right\rangle :i=1,\dots,n\right\} $
(called connection vector) satisfying the condition of Lemma~\ref{lem:joint-of-two}
for each $i=1,\dots,n$, the condition 
\begin{equation}
\begin{array}{l}
\left\langle V_{k}W_{k}\right\rangle =1-\left|\left\langle V_{k}\right\rangle -\left\langle W_{k}\right\rangle \right|,\\
\left\langle V_{i}W_{i}\right\rangle \ge\left|1-\left|\left\langle V_{k}\right\rangle -\left\langle W_{k}\right\rangle \right|\right|,\, i\ne k
\end{array}\label{eq:restriction}
\end{equation}
for some index $k$, and the inequality $S\le M$, then there exist
another connection vector $\left\{ \left\langle V_{i}W_{i}\right\rangle :i=1,\dots,n\right\} $
such that a joint of the given observable pairs having these connection
expectations exists and satisfies the equation $S=M$.\end{lem}
\begin{proof}
Consider the set of all connection vectors $\left\{ \left\langle V_{i}W_{i}\right\rangle :i=1,\dots,n\right\} $
satisfying the condition of Lemma~\ref{lem:joint-of-two} for each
pair $\left(V_{i},W_{i}\right)$, $i=1,\dots,n$, and the condition
(\ref{eq:restriction}). Within this set we have the assumed configuration
satisfying $S\le M$ and the configuration $\left\langle V_{i}W_{i}\right\rangle =1-\left|\left\langle V_{i}\right\rangle -\left\langle W_{i}\right\rangle \right|$,
$i=1,\dots,n$, which yields $S>M$ due to (\ref{eq:simpleform})
not being satisfied. Since $S$ is a continuous function of the vector
$\left\{ \left\langle V_{i}W_{i}\right\rangle :i=1,\dots,n\right\} $
and the conjunction of the condition of Lemma~\ref{lem:joint-of-two}
and condition (\ref{eq:restriction}) defines a connected set of connection
vectors, there exist a configuration $\left\{ \left\langle V_{i}W_{i}\right\rangle :i=1,\dots,n\right\} $
satisfying the condition of Lemma~\ref{lem:joint-of-two} and (\ref{eq:restriction})
with equality $S=M$. For this configuration, expanding the definitions
of $S$ and $M$ we obtain
\begin{equation}
\so\left(\left\langle V_{i}W_{i\oplus1}\right\rangle :i=1,\dots,n\right)+\sum_{i=1}^{n}\left\langle V_{i}W_{i}\right\rangle =2n-2.\label{eq:cont-lemma-1}
\end{equation}
Also, since (\ref{eq:restriction}) implies condition 1 of Lemma~\ref{lem:interesting-trivial}
we obtain from Lemma~\ref{lem:expand_s0_s1} 
\begin{align}
\se\left(\left\langle V_{i}W_{i}\right\rangle :i=1,\dots,n\right) & =\sum_{i=1}^{n}\left\langle V_{i}W_{i}\right\rangle ,\label{eq:cont-lemma-2}\\
\so\left(\left\langle V_{i}W_{i}\right\rangle :i=1,\dots,n\right) & =\sum_{i\ne k}\left\langle V_{i}W_{i}\right\rangle -\left(1-\left|\left\langle V_{k}\right\rangle -\left\langle W_{k}\right\rangle \right|\right).\label{eq-cont-lemma-3}
\end{align}
Now (\ref{eq:cont-lemma-1}) and (\ref{eq:cont-lemma-2}) imply
\[
\so\left(\left\langle V_{i}W_{i\oplus1}\right\rangle :i=1,\dots,n\right)+\se\left(\left\langle V_{i}W_{i}\right\rangle :i=1,\dots,n\right)=2n-2
\]
and Lemma~\ref{lem:joint-of-two} applied to (\ref{eq:cont-lemma-2})
and (\ref{eq-cont-lemma-3}) yield 
\begin{align*}
 & \se\left(\left\langle V_{i}W_{i\oplus1}\right\rangle :i=1,\dots,n\right)+\so\left(\left\langle V_{i}W_{i}\right\rangle :i=1,\dots,n\right)\\
 & \le\sum_{i=1}^{n}\left(1-\left|\left\langle V_{i}\right\rangle -\left\langle W_{i\oplus1}\right\rangle \right|\right)+\sum_{i\ne k}\left(1-\left|\left\langle V_{i}\right\rangle -\left\langle W_{i}\right\rangle \right|\right)-\left(1-\left|\left\langle V_{k}\right\rangle -\left\langle W_{k}\right\rangle \right|\right)\\
 & =2n-2+\underbrace{\left|\left\langle V_{k}\right\rangle -\left\langle W_{k}\right\rangle \right|-\sum_{i=1}^{n}\left|\left\langle V_{i}\right\rangle -\left\langle W_{i\oplus1}\right\rangle \right|-\sum_{i\ne k}\left|\left\langle V_{i}\right\rangle -\left\langle W_{i}\right\rangle \right|.}_{\le0\,\text{(by the triangle inequality \eqref{eq:cyclic-triangle})}}
\end{align*}
Thus, by By Lemma~\ref{lem:s0_s1_split} and Theorem~\ref{thm:A1-An},
a joint exists.
\end{proof}
Now we are ready to prove the main result. It suffices to prove the
following partial statement (Theorem~\ref{thm:measure-corr} and
hence also Theorem~\ref{thm:measure} follow immediately):
\begin{thm}
\label{thm:measure-correlations}For each $i=1,\dots,n$ ($n\ge2$),
let the distribution of a pair $\left(V_{i},W_{i\oplus1}\right)$
of $\pm1$-valued random variables be given. If the main criterion
(\ref{eq:master}) does not hold, then the maximum possible value
of
\[
S=\sum_{i=1}^{n}\left\langle V_{i}W_{i}\right\rangle 
\]
over all possible joints of the given pairs is given by
\begin{equation}
M=2n-2-\so\left(\left\langle V_{i}W_{i\oplus1}\right\rangle :i=1,\dots,n\right)\ge n-2.\label{eq:measure}
\end{equation}
\end{thm}
\begin{proof}
A larger value cannot be reached since we have
\begin{align*}
 & 2n-2\\
\text{(Theorem\,\ref{thm:A1-An})} & \ge\so\left(\left\langle V_{i}W_{i\oplus1}\right\rangle ,\left\langle V_{i}W_{i}\right\rangle :i=1,\dots,n\right)\\
\text{(Lemma\,\ref{lem:s0_s1_split})} & \ge\so\left(\left\langle V_{i}W_{i\oplus1}\right\rangle :i=1,\dots,n\right)+\se\left(\left\langle V_{i}W_{i}\right\rangle :i=1,\dots,n\right)\\
 & \ge\so\left(\left\langle V_{i}W_{i\oplus1}\right\rangle :i=1,\dots,n\right)+S.
\end{align*}
To show that this value can be reached, without loss of generality
(by Lemmas~\ref{lem:WLOG} and \ref{lem:WLOG-permutation}) assume
\begin{equation}
\left\langle V_{i}W_{i\oplus1}\right\rangle \ge\left|\left\langle V_{n}W_{1}\right\rangle \right|,\quad i=1,\dots,n-1,\label{eq:obs-cond-for-all-cases}
\end{equation}
and since the criterion (\ref{eq:master}) is not satisfied, by Lemma~\ref{lem:trivialcase},
the maximal connection expectations $1-\left|\left\langle V_{i}\right\rangle -\left\langle W_{i}\right\rangle \right|$
satisfy for some $k$ the condition 
\begin{equation}
1-\left|\left\langle V_{i}\right\rangle -\left\langle W_{i}\right\rangle \right|\ge\left|1-\left|\left\langle V_{k}\right\rangle -\left\langle W_{k}\right\rangle \right|\right|,\quad i\ne k.\label{eq:conn-cond-for-all-cases}
\end{equation}
We will show that a specific configuration $\left\{ \left\langle V_{i}W_{i}\right\rangle :i=1,\dots,n\right\} $
exists for which $S=M$ and for which a joint of the pairs $\left(V_{i},W_{i\oplus1}\right)$,
$i=1,\dots,n$, exists. This is shown separately for four cases that
exhaust all possible situations under the present assumptions. The
four cases are defined as follows:

\emph{
\[
\xymatrix{\forall i:\left|\left\langle V_{i}\right\rangle +\left\langle W_{i}\right\rangle \right|-1\le\left|1-\left|\left\langle V_{k}\right\rangle -\left\langle W_{k}\right\rangle \right|\right|?\ar[d]_{No}\ar[dr]_{Yes}\hspace{-8ex}\\
\textnormal{Case 1} & 1-\left|\left\langle V_{k}\right\rangle -\left\langle W_{k}\right\rangle \right|\ge0?\ar[d]_{No}\ar[dl]_{Yes}\\
\left\langle V_{n}W_{1}\right\rangle \le1-\left|\left\langle V_{k}\right\rangle -\left\langle W_{k}\right\rangle \right|?\ar[d]_{No}\ar[dr]_{Yes} & \textnormal{Case 2}\\
\textnormal{Case 3} & \textnormal{Case 4}
}
\]
}

\textbf{Case 1.} We have for some index $j$, 
\begin{equation}
\big|\langle V_{j}\rangle+\langle W_{j}\rangle\big|-1>\left|1-\left|\left\langle V_{k}\right\rangle -\left\langle W_{k}\right\rangle \right|\right|.\label{eq:case1-thm-measure}
\end{equation}
In this case, we define $\left\langle V_{i}W_{i}\right\rangle =1-\left|\left\langle V_{i}\right\rangle -\left\langle W_{i}\right\rangle \right|$
for all $i\ne j$ and choose an arbitrary value $\langle V_{j}W_{j}\rangle\ge\big|\langle V_{j}\rangle+\langle W_{j}\rangle\big|-1$
from the range given by Lemma~\ref{lem:chain-joint-range} (intersected
with the range given by Lemma~\ref{lem:joint-of-two}). Since a joint
exists by Lemma~\ref{lem:chain-joint-range} for this configuration,
we obtain $S\le M$. Furthermore, since (\ref{eq:conn-cond-for-all-cases})
and (\ref{eq:case1-thm-measure}) imply that condition (\ref{eq:restriction})
is satisfied for this configuration, Lemma~\ref{lem:S=00003DM} implies
that a joint exists for some configuration with $S=M$.

\textbf{Case 2.} We have 
\begin{gather}
1-\left|\left\langle V_{k}\right\rangle -\left\langle W_{k}\right\rangle \right|<0,\nonumber \\
\big|\langle V_{j}\rangle+\langle W_{j}\rangle\big|-1\le\left|1-\left|\left\langle V_{k}\right\rangle -\left\langle W_{k}\right\rangle \right|\right|,\quad j\ne k.\label{eq:case2-thm-measure}
\end{gather}
In this case, we choose an index $j\ne k$ and define $\langle V_{j}W_{j}\rangle=\left|1-\left|\left\langle V_{k}\right\rangle -\left\langle W_{k}\right\rangle \right|\right|=-\left(1-\left|\left\langle V_{k}\right\rangle -\left\langle W_{k}\right\rangle \right|\right)$
and $\left\langle V_{i}W_{i}\right\rangle =1-\left|\left\langle V_{i}\right\rangle -\left\langle W_{i}\right\rangle \right|$
for $i\ne j$. These values satisfy condition (\ref{eq:restriction})
by (\ref{eq:conn-cond-for-all-cases}) and the condition of Lemma~\ref{lem:joint-of-two}
by (\ref{eq:conn-cond-for-all-cases}) and (\ref{eq:case2-thm-measure}).
Now $\langle V_{j}W_{j}\rangle+\left\langle V_{k}W_{k}\right\rangle =0$
implying that $S\le n-2\le M$ and so by Lemma~\ref{lem:S=00003DM}
a joint exists for a configuration with $S=M$.

\textbf{Case 3.} We have 
\begin{gather}
0\le1-\left|\left\langle V_{k}\right\rangle -\left\langle W_{k}\right\rangle \right|<\left\langle V_{n}W_{1}\right\rangle ,\nonumber \\
\left|\left\langle V_{i}\right\rangle +\left\langle W_{i}\right\rangle \right|-1\le1-\left|\left\langle V_{k}\right\rangle -\left\langle W_{k}\right\rangle \right|,\quad i=1,\dots,n.\label{eq:case3-thm-measure}
\end{gather}
In this case, we define $\left\langle V_{i}W_{i}\right\rangle =1-\left|\left\langle V_{k}\right\rangle -\left\langle W_{k}\right\rangle \right|\ge0$
for all $i=1,\dots,n$ and condition (\ref{eq:restriction}) is satisfied
trivially. These values also satisfy the condition of Lemma~\ref{lem:joint-of-two}
by (\ref{eq:conn-cond-for-all-cases}) and (\ref{eq:case3-thm-measure}).
Furthermore, by (\ref{eq:obs-cond-for-all-cases}) and (\ref{eq:case3-thm-measure}),
all product expectations $\left\{ \left\langle V_{i}W_{i\oplus1}\right\rangle ,\left\langle V_{i}W_{i}\right\rangle :i=1,\dots,n\right\} $
are nonnegative and the smallest value among them is $\left\langle V_{i}W_{i}\right\rangle =1-\left|\left\langle V_{k}\right\rangle -\left\langle W_{k}\right\rangle \right|$.
So, by Lemma~\ref{lem:expand_s0_s1}, we can expand 
\begin{align*}
 & \so\left(\left\langle V_{i}W_{i\oplus1}\right\rangle ,1-\left|\left\langle V_{k}\right\rangle -\left\langle W_{k}\right\rangle \right|:i=1,\dots,n\right)\\
 & =\sum_{i=1}^{n}\left\langle V_{i}W_{i\oplus1}\right\rangle +(n-2)\left(1-\left|\left\langle V_{k}\right\rangle -\left\langle W_{k}\right\rangle \right|\right)\\
 & \le n+(n-2)\le2n-2.
\end{align*}
Thus, a joint exists by Theorem~\ref{thm:A1-An} which implies that
$S\le M$ and so by Lemma~\ref{lem:S=00003DM} a joint exists for
a configuration with $S=M$.

\textbf{Case 4.} We have 
\begin{gather}
\left\langle V_{n}W_{1}\right\rangle \le1-\left|\left\langle V_{k}\right\rangle -\left\langle W_{k}\right\rangle \right|\ge0,\label{eq:case4-thm-measure}\\
\left|\left\langle V_{i}\right\rangle +\left\langle W_{i}\right\rangle \right|-1\le1-\left|\left\langle V_{k}\right\rangle -\left\langle W_{k}\right\rangle \right|,\quad i=1,\dots,n.\nonumber 
\end{gather}
In this case, we define $\left\langle V_{i}W_{i}\right\rangle =1-\left|\left\langle V_{k}\right\rangle -\left\langle W_{k}\right\rangle \right|\ge\left|\left\langle V_{i}\right\rangle +\left\langle W_{i}\right\rangle \right|-1$
for $i\ne k$ and (\ref{eq:conn-cond-for-all-cases}) implies that
these values satisfy the condition of Lemma~\ref{lem:joint-of-two}.
Thus, a joint exists by Lemma~\ref{lem:chain-joint-range} for
\begin{equation}
\left\langle V_{k}W_{k}\right\rangle =\min\begin{cases}
2n-2-\so\big(\left\langle V_{i}W_{i\oplus1}\right\rangle :i=1,\dots,n;\left\langle V_{i}W_{i}\right\rangle :i\ne k\big),\\
1-\left|\left\langle V_{k}\right\rangle -\left\langle W_{k}\right\rangle \right|.
\end{cases}\label{eq:case1-conn-k}
\end{equation}

Assume first that the top expression in (\ref{eq:case1-conn-k}) is
the minimum. Since here $\left\langle V_{n}W_{1}\right\rangle $ is
by (\ref{eq:obs-cond-for-all-cases}) and (\ref{eq:case4-thm-measure})
the smallest argument to $\so\left(\dots\right)$ in (\ref{eq:case1-conn-k})
and all other arguments are nonnegative Lemma~\ref{lem:expand_s0_s1}
yields
\begin{align*}
 & \so\big(\left\langle V_{i}W_{i\oplus1}\right\rangle :i=1,\dots,n;\left\langle V_{i}W_{i}\right\rangle :i\ne k\big)\\
 & =\sum_{i=1}^{n-1}\left\langle V_{i}W_{i\oplus1}\right\rangle -\left\langle V_{n}W_{1}\right\rangle +(n-1)\left(1-\left|\left\langle V_{k}\right\rangle -\left\langle W_{k}\right\rangle \right|\right)\\
 & =\so\left(\left\langle V_{i}W_{i\oplus1}\right\rangle :i=1,\dots,n\right)+(n-1)\left(1-\left|\left\langle V_{k}\right\rangle -\left\langle W_{k}\right\rangle \right|\right)\\
 & =\so\left(\left\langle V_{i}W_{i\oplus1}\right\rangle :i=1,\dots,n\right)+S-\left\langle V_{k}W_{k}\right\rangle .
\end{align*}
 Substituting this in (\ref{eq:case1-conn-k}) yields
\[
S=2n-2-\so\left(\left\langle V_{i}W_{i\oplus1}\right\rangle :i=1,\dots,n\right)=M.
\]

Suppose then that the bottom expression in (\ref{eq:case1-conn-k})
is the minimum. Then, we have 
\[
\left\langle V_{i}W_{i}\right\rangle =1-\left|\left\langle V_{k}\right\rangle -\left\langle W_{k}\right\rangle \right|\ge0,\quad i=1,\dots,n
\]
and condition (\ref{eq:restriction}) is satisfied and since a joint
exists, it follows $S\le M$. But then by Lemma~\ref{lem:S=00003DM}
a joint exists for a configuration with $S=M$.
\end{proof}

\section{Conclusion}

We have derived the formula (\ref{eq:conjectured cntx}) for contextuality
measure in cyclic systems with binary variables, previously conjectured
based on the principles laid out in Refs. \cite{DKL2015,KDL2015}.
The measure is based on the principled comparison of the minimal values
$\Delta_{0}$ and $\Delta_{\min}$ of 
\[
\Delta=\sum_{i=1}^{n}\Pr\left[V_{i}\ne W_{i}\right],
\]
with $\Delta_{0}$ computed across all couplings for the individual
pairs $\left\{ V_{i},W_{i}\right\} $ (i.e., by minimizing each $\Pr\left[V_{i}\ne W_{i}\right]$
separately), and with $\Delta_{\min}$ computed across all couplings
for the entire set of the random variables $\left\{ V_{i},W_{i}:i=1,\ldots,n\right\} $.

A logical consequence of this measure is the criterion (necessary
and sufficient condition) for contextuality (\ref{eq:conjectured criterion}).
In this paper, however, we proved (\ref{eq:conjectured criterion})
by showing its equivalence to the criterion (\ref{eq:proved criterion})
derived in Ref. \cite{KDL2015}, and used this equivalence to derive
(\ref{eq:conjectured cntx}).

\subsubsection*{Acknowledgments}

This work is supported by NSF grant SES-1155956, AFOSR grant FA9550-14-1-0318, and A. von Humboldt Foundation.
The authors benefited from collaboration with Acacio de Barros, Gary
Oas, and Jan-Åke Larsson.

\end{document}